\documentclass[submission,copyright,creativecommons]{eptcs}
\providecommand{\event}{Refinement Workshop at FM'13}
\usepackage{latexsym,amsmath,amssymb,listings}

\catcode`\@=11
\SetMathAlphabet{\mathrm}{normal}{\encodingdefault}{\rmdefault}{m}{n}%
\SetMathAlphabet{\mathbf}{normal}{\encodingdefault}{\rmdefault}{bx}{n}%
\SetMathAlphabet{\mathsf}{normal}{\encodingdefault}{\sfdefault}{m}{n}%
\DeclareSymbolFont{italics}{\encodingdefault}{\rmdefault}{m}{it}%
\DeclareSymbolFontAlphabet{\mathrm}{operators}
\DeclareSymbolFontAlphabet{\mathit}{letters}
\DeclareSymbolFontAlphabet{\mathcal}{symbols}
\def\@setmcodes#1#2#3{{\count0=#1 \count1=#3
        \loop \global\mathcode\count0=\count1 \ifnum \count0<#2
        \advance\count0 by1 \advance\count1 by1 \repeat}}
\@setmcodes{`A}{`Z}{"7\hexnumber@\symitalics41}%
\@setmcodes{`a}{`z}{"7\hexnumber@\symitalics61}%
\newcommand{\normalfootnotesize}{%
   \@setfontsize\normalfootnotesize\@viiipt{9.5}%
   \abovedisplayskip 6\p@ \@plus2\p@ \@minus4\p@
   \abovedisplayshortskip \z@ \@plus\p@
   \belowdisplayshortskip 3\p@ \@plus\p@ \@minus2\p@
   \def\@listi{\leftmargin\leftmargini
               \topsep 3\p@ \@plus\p@ \@minus\p@
               \parsep 2\p@ \@plus\p@ \@minus\p@
               \itemsep \parsep}%
   \belowdisplayskip \abovedisplayskip
}
\catcode`\@=12

\newcommand{\KW}[1]{\mathop{\sf #1}}

\newcommand{\TRUE}{\KW{true}}
\newcommand{\FALSE}{\KW{false}}
\newcommand{\ABORT}{\KW{abort}}
\newcommand{\SKIP}{\KW{skip}}
\newcommand{\RAISE}{\KW{raise}}
\newcommand{\STOP}{\KW{stop}}
\newcommand{\MAGIC}{\KW{magic}}
\newcommand{\semi}{\mathbin{;}}
\newcommand{\meet}{\sqcap}
\newcommand{\join}{\sqcup}
\newcommand{\SEMI}{\mathbin{;\!;}}
\newcommand{\TRY}{\KW{try}}
\newcommand{\CATCH}{\KW{catch}}
\newcommand{\FINALLY}{\KW{finally}}

\newcommand{\IF}{\KW{if}}
\newcommand{\THEN}{\KW{then}}
\newcommand{\ELSE}{\KW{else}}
\newcommand{\WHILE}{\KW{while}}
\newcommand{\DO}{\KW{do}}

\newcommand{\DEF}{\KW{def}}
\newcommand{\VAL}{\KW{val}}
\newcommand{\con}{\wedge}
\newcommand{\dis}{\vee}
\newcommand{\imp}{\Rightarrow}
\newcommand{\cons}{\Leftarrow}
\newcommand{\fun}{\rightarrow}
\newcommand{\rel}{\leftrightarrow}
\renewcommand{\dot}{\mathop{.}}
\newcommand{\defeq}{\left.\widehat{=}\right.}
\newcommand{\refby}{\sqsubseteq}
\newcommand{\ID}{\KW{id}}
\newcommand{\DIV}{\KW{div}}
\newcommand{\TOTAL}[3]{[\hspace{-.15em}[#1]\hspace{-.15em}]\,#2\,[\hspace{-.15em}[#3]\hspace{-.15em}]}
\newcommand{\PARTIAL}[3]{\langle\hspace{-.2em}|#1|\hspace{-.2em}\rangle\,#2\,\langle\hspace{-.2em}|#3|\hspace{-.2em}\rangle}
\newcommand{\trefby}{\sqsubseteq}

\newcommand{\prefby}{\mathbin{\raisebox{.15em}{$\sqsubset$}\hspace{-.73em}\raisebox{-.3em}{\small$\sim$}}}
\newcommand{\RESULT}{\KW{result}}

\newcommand{\eq}{\equiv}

\newcommand{\CLASS}{\KW{class}}
\newcommand{\METHOD}{\KW{method}}

\newcommand{\Int}{\KW{int}}
\newcommand{\INVARIANT}{\KW{invariant}}
\newcommand{\NEW}{\KW{new}}
\newcommand{\LOOPINV}{\KW{loop\ invariant}}

\title{On a New Notion of Partial Refinement}
\author{Emil Sekerinski
  \institute{McMaster University\\ Hamilton, Canada}
  \email{emil@mcmaster.ca}
\and Tian Zhang
  \institute{McMaster University\\ Hamilton, Canada}
  \email{zhangt26@mcmaster.ca}
}
\def\titlerunning{Partial Refinement}
\def\authorrunning{E. Sekerinski \& T. Zhang}
\begin{document}
\maketitle

\newenvironment{eqnarr}
    {\[\begin{array}{@{}lllllllllll@{}}}
    {\end{array}\]}
\newtheorem{definition}{Definition}
\newtheorem{theorem}{Theorem}
\newenvironment{proof}{\par\noindent\em Proof.}{\par}

\begin{abstract}

Formal specification techniques allow expressing idealized specifications, which abstract from restrictions that may arise in implementations. However, partial implementations are universal in software development due to practical limitations. Our goal is to contribute to a method of program refinement that allows for partial implementations. For programs with a normal and an exceptional exit, we propose a new notion of partial refinement which allows an implementation to terminate exceptionally if the desired results cannot be achieved, provided the initial state is maintained. Partial refinement leads to a systematic method of developing programs with exception handling.

\end{abstract}

\section{Introduction}

In software development, specifications are meant to be concise by stating abstractly only the intention of a program rather than elaborating on a possible implementation. However, practical restrictions can prevent idealized specifications from being fully implemented. In general, there are three sources of {\em partiality} in implementations: there may be inherent limitations of the implementation, some features may intentionally not (yet) be implemented, or there may be a genuine fault.

As an example of inherent limitations of an implementation, consider a class for the analysis of a collection of integers. The operations are initialization, inserting an integer, and summing all its elements. Assume that $\Int$ is a type for machine-representable integers, bounded by $MIN$ and $MAX$, and machine arithmetic is bounded, i.e. an overflow caused by arithmetic operations on $\Int$ is detected and raises an exception, as available in x86 assembly language~\cite{Intel13Assembly} and .NET~\cite{MS13OverflowException}.
We define:

\begin{center}
$\begin{array}{l}
\CLASS~IntCollection\\
\quad \mathop{\sf bag}(\Int)\,b\\
\quad \INVARIANT~inv: \forall\, x \in b \cdot MIN \leq x \leq MAX \\
\quad \METHOD~init()\\
\quad \quad b := []\\
\quad \METHOD~insert(n : \Int)\\
\quad \quad b := b + [n] \\
\quad \METHOD~sum(): \Int\\
\quad \quad \RESULT := \sum b \\
\end{array}$
\end{center}
This specification allows an unbounded number of machine-representable integers to be stored in the abstract bag $b$, which is an unordered collection that, unlike a set, allows duplication of elements, The empty bag is written as $[]$, the bag consisting only of a single $n$ as $[n]$, union of bags $b, b'$ as $b + b'$, and the sum of all elements of bag $b$ as $\sum b$. (A model of bags is a function from elements to their number of occurrences in the bag.) However, in an implementation, method $init$ (object initialization) and $insert$ may fail due to memory exhaustion at heap allocation (and raise an exception), and method $sum$ may fail due to overflow of the result (and raise an exception). Even if the result of $sum$ is machine-representable, the iterative computation of the sum in a particular order may still overflow. Hence no realistic implementation of this class can be faithful. Still, it is a useful specification because of its clarity and brevity. Obviously, using mathematical integers instead of machine-representable integers in the specification would make implementations even ``more partial''.

The second source of partiality is intentionally missing features. The evolutionary development of software consists of anticipating and performing \emph{extensions} and~\emph{contractions}~\cite{Parnas78ExtensionContraction}. Anticipated extensions can be expressed by features that are present but not yet implemented. Contractions lead to obsolete features that will eventually be removed. In both cases, implementations of features may be missing. A common practice is to raise an exception if a feature is not yet implemented. Here is a recommendation from the Microsoft Developer Documentation for .NET~\cite{MS13NotImplementedException}:

\begin{center}
\begin{tabular}{c}
\begin{lstlisting}[language=C++,columns=fullflexible,identifierstyle=\itshape,commentstyle=]
static void FutureFeature()
{
   // Not developed yet.
   throw new NotImplementedException();
}
\end{lstlisting}
\end{tabular}
\end{center}

The third source of partiality is genuine faults. These may arise from the use of software layers that are themselves faulty (operating system, compilers, libraries), from faults in the hardware (transient or permanent), or from errors in the design (errors in the correctness argument, incorrect hypothesis about the abstract machine).

Our goal is to contribute to a method of program refinement that allows for partial implementations that guarantee ``safe" failure if the desired outcome cannot be computed. To this end, we consider program statements with a one entry and two exits, a normal and an exceptional exit. Specifications are considered to be abstract programs and implementations are  considered to be concrete programs. For program statements $S$ (the specification) and $T$ (the implementation), the total refinement of $S$ by $T$ means that either $T$ terminates normally and establishes a postcondition that $S$ may also establish on normal termination, or $T$ terminates exceptionally and establishes a postcondition that $S$ may also establish on exceptional termination. As a relaxation, {\em partial refinement} allows $T$ additionally to terminate exceptionally provided the initial state is preserved. The intention is that the implementation $T$ tries to meet specification $S$, but if it cannot do so, $T$ can fail safely by not changing the state and terminating exceptionally. When applying partial refinement to data refinement, an implementation that cannot meet the specification and fails may still change the state as long as the change is not visible in the specification.


The exception handling of the Eiffel programming language provided the inspiration for partial refinement~\cite{Meyer97OOSoftwareConstruction}: in Eiffel, each method has one entry and two exits, a normal and an exceptional exit, but is specified by a single precondition and single postcondition only. The normal exit is taken if the desired postcondition is established and the exceptional exit is taken if the desired postcondition cannot be established, thus allowing partial implementations. Other approaches to  exceptions rely on specifying one postcondition for normal exit and one for each exceptional exit (in case there is more than one exceptional exit). Compared to that, the Eiffel approach is simpler in requiring a single postcondition and more general in that an valid outcome is always possible, even in presence of unanticipated failure. Compared to Eiffel, our notion of partial refinement imposes the additional constraint that in case of exceptional exit, the (abstract) state must be preserved. In Eiffel, the postcondition must be evaluated at run-time to determine if the normal or exceptional exit it taken (and hence must be efficiently computable). Partial refinement is more general  in the sense that it does not require that a postcondition to be evaluated at run-time, as long as faults are detected at run-time in some way. In our earlier work we introduced a new notion of {\em partial correctness}
~\cite{SekerinskiZhang11PartialCorrectness}, which also models the programs that can fail safely. (This notion of partial correctness is different from relaxing total correctness by not guaranteeing termination~\cite{JacobsGries85GeneralCorrectness}.) Partial refinement is related to this notion of partial correctness in the same sense as total refinement is related to total correctness.

Retrenchment also addresses the issue of partial implementations, but with statements with one entry and one exit~\cite{BanachaPoppletonJeskeaStepney07Retrenchment}: the refinement of each operation of a data type requires a \emph{within} and a \emph{concedes} relation that restrict the initial states and widen the possible final states of an implementation. Compared to retrenchment, partial refinement does not require additional relations to be specified. Partial refinement is more restrictive in the sense that initial states cannot be restricted and the final state cannot be widened on normal termination. In partial refinement, the caller is notified through an exception of the failure; retrenchment is a design technique that is independent of exception handling.

The term partial refinement has been introduced with a different meaning in~\cite{Back81CorrectRefinement} for statements with one exit. There, partial refinement allows the domain of termination to be reduced. In the present work, partial refinement requires termination, either on the normal or exceptional exit. Partial refinement in \cite{JeffordsHeitmeyerArcherLeonard09CorrectFaultTolerantSystems} refers to transition system refinement with a partial refinement relation; the approach is to construct fault-tolerant systems in two phases, first with an idealized specification and then adding fault-tolerant behaviour. Here we use partial refinement specifically for statements with two exits.

We define statements by predicate transformers, as they are expressive in that they can describe blocking, abortion, demonic nondeterminism, and angelic nondeterminism, see e.g.~Back and von Wright~\cite{BackVonWright98RefinementCalculus} for a comprehensive treatment; that line of work considers statements with one entry and one exit. In the approach to exception handling by Cristian~\cite{Cristian84CorrectRobustPrograms} statements have one entry and multiple exits (one of those being the normal one) and are defined by a set of predicate transformers, one for each exit. As pointed out by King and Morgan~\cite{KingMorgan95ExitsInRefinementCalculus}, this disallows nondeterminism, which precludes the use of the language for specification and design; their solution is to use a single predicate transformer with one postcondition for each exit instead and define total refinement as a generalization of refinement of statements with one exit. (For example, $\SKIP \meet \RAISE$ terminates, but is neither guaranteed to terminate normally nor exceptionally; hence its termination cannot be captured by a predicate transformer for each exit.) Leino and Snepscheut justify a very similar definition of statements by predicate transformers with two postconditions in terms of a trace semantics, but do not consider refinement~\cite{LeinoSnepscheut94SemanticsExceptions}. Here, we identify a statement $S$ with its (higher-order) predicate transformer following~\cite{BackVonWright98RefinementCalculus}, instead of introducing a function $wp(S)$ in Dijkstra's style~\cite{Dijkstra75GuardedCommands}. Watson uses total refinement to further study the refinement in the style of King and Morgan~\cite{Watson02RefiningExceptions}. With total refinement, an implementation has to meet its specifications, but may use raising and catching exceptions as an additional control structure. Our definition of partial refinement on statements with two exits is in terms of total refinement, but allows implementations to fail safely. Compared to earlier work on refinement with exceptions, we consider undefinedness in expressions and give corresponding (more involved) verification rules. King and Morgan trace the use of predicate transformers with multiple postconditions to Back and Karttunen~\cite{BackKarttunen83PredicateTransformerMultipleExits}. 

\paragraph{Overview.} The next two sections give the formal definitions of statements with two exits using predicate transformers. Section~\ref{sec:toco} and~\ref{sec:paco} introduce the notions of total correctness and partial correctness for statements with two exits. Section~\ref{sec:loop} gives verification rules for loops. In Section~\ref{sec:partial} we give the formal definition of partial refinement, and show its applications through examples. Section~\ref{sec:conclusion} concludes our contribution.

All the theorems and examples have been checked with the Isabelle proof assistant\footnote{The Isabelle proof files are available at \url{http://www.cas.mcmaster.ca/~zhangt26/REFINE/}.}, so we allow ourselves to omit the proofs.

\section{Statements with Two Exits}
\label{sec:basic}

We begin by formally defining statements with a one entry and two exits, a {\em normal} and an {\em exceptional} exit.  A statement either {\em succeeds} (terminates normally), {\em fails} (terminates exceptionally), {\em aborts} (is out of control and may not terminate at all), or {\em blocks} (refuses to execute). Following~\cite{BackVonWright98RefinementCalculus}, a statement $S$ is identified with its predicate transformer, but taking two postconditions as arguments: $S\,(q, r)$ is the weakest precondition such that either $S$ succeeds with postcondition $q$ or $S$ fails with postcondition $r$.

A {\em state predicate} of type ${\cal P} \Sigma$ is a function from state space $\Sigma$ to $Bool$, i.e.~${\cal P} \Sigma = \Sigma \fun Bool$. A {\em relation} of type $\Delta \leftrightarrow \Sigma$ is a function from the initial state space $\Delta$ to a state predicate over the final state space $\Sigma$, i.e. $\Delta \leftrightarrow \Sigma = \Delta \fun {\cal P} \Sigma$; we allow the initial and final state spaces to be different. It is isomorphic to ${\cal P} (\Delta \times \Sigma)$, but allows the test if $(a, b)$ is in relation $r$ to be simply written as $r\,a\,b$.

A {\em predicate transformer} is a function from a normal postcondition of type ${\cal P} \Psi$ and an exceptional postcondition of type ${\cal P} \Omega$, to a precondition of type ${\cal P} \Delta$, i.e. of type ${\cal P} \Psi \times {\cal P} \Omega \fun {\cal P} \Delta$, for types $\Psi, \Omega, \Delta$. In the rest of the paper, we leave the types out if they can be inferred from the context.

On state predicates, conjunction $\con$, disjunction $\dis$, implication $\imp$, consequence $\cons$, and negation $\neg$ are defined by the pointwise extension of the corresponding operations on $Bool$, e.g. $(p \con q) \sigma \defeq p\,\sigma \con q\,\sigma$. The entailment ordering $\leq$ is defined by universal implication, $p \leq q \defeq \forall\, \sigma \dot p\,\sigma \imp q\,\sigma$. The state predicates $\TRUE$ and $\FALSE$ represent the universally $\TRUE$ respectively $\FALSE$ predicates. Predicate transformer $S$ is {\em monotonic} if $q \leq q' \con r \leq r' \imp S\,(q, r) \leq S\,(q', r')$. A {\em statement} is a monotonic predicate transformer.


We define some basic predicate transformers: $\ABORT$ is completely unpredictable and may terminate normally or exceptionally in any state or may not terminate at all; $\STOP$, also known as $\MAGIC$, miraculously guarantees any postcondition by blocking execution; $\SKIP$ changes nothing and succeeds whereas $\RAISE$ changes nothing and fails. The {\em sequential composition} $S \semi T$ continues with $T$ only if $S$ succeeds whereas the {\em exceptional composition} $S \SEMI T$ continues with $T$ only if $S$ fails. The notation suggests that $\semi$ is sequential composition on the first exit and $\SEMI$ is sequential composition on the second exit. The {\em demonic choice} $S \meet T$ establishes a postcondition only if both $S$ and $T$ do. The {\em angelic choice} $S \join T$ establishes a postcondition if either $S$ or $T$ does:
\[\renewcommand{\arraystretch}{1.2}\begin{array}{lcl@{\qquad\qquad}lcl}
   \ABORT\,(q, r) & \defeq & \FALSE & (S \semi T)\,(q, r) & \defeq & S\,(T\,(q, r), r) \\
   \STOP\,(q, r)  & \defeq & \TRUE  & (S \SEMI T)\,(q, r) & \defeq & S\,(q, T\,(q, r)) \\
   \SKIP\,(q, r)  & \defeq & q      & (S \meet T)\,(q, r) & \defeq & S\,(q, r) \con T\,(q, r) \\
   \RAISE\,(q, r) & \defeq & r      & (S \join T)\,(q, r) & \defeq & S\,(q, r) \dis T\,(q, r)
\end{array}\]
Predicate transformers $\ABORT$, $\STOP$, $\SKIP$, and $\ABORT$ are monotonic and hence statements. Operators $\semi$, $\SEMI$, $\meet$, $\join$ preserve monotonicity. Sequential composition is associative and has $\SKIP$ as unit, giving rise to a monoid structure. Dually, exceptional composition is associative and has $\RAISE$ as unit, giving rise to another monoid structure. Further properties of functions with two arguments, with application to semantics of exceptions, are studied in~\cite{LeinoSnepscheut94SemanticsExceptions, ManoharLeino95ConditionalComposition}.

The {\em total refinement ordering} $\refby$ on predicate transformers with two arguments is defined by universal entailment~\cite{KingMorgan95ExitsInRefinementCalculus}:
\[\renewcommand{\arraystretch}{1.2}\begin{array}{rcl}
  S \refby T & \defeq & \forall\, q, r \dot S\,(q, r) \leq T\,(q, r)
\end{array}\]
Both sequential and exceptional composition are monotonic in both arguments with respect to the total refinement ordering. Statements with the total refinement ordering form a complete lattice, with $\ABORT$ as the bottom element, $\STOP$ as the top element, $\meet$ as the meet operator, and $\join$ as the join operator. These properties are similar to those of predicate transformers with one argument~\cite{BackVonWright98RefinementCalculus}, except that here we have two monoid structures.

We introduce statements for inspecting and modifying the state. Let $u,v$ be state predicates. For predicate transformers with one argument, the assumption $[u]$  does nothing if $u$ holds and blocks otherwise, and the assertions $\{u\}$ does nothing if $u$ hold and aborts otherwise. For predicate transformers with two arguments, the {\em assumption} $[u, v]$ succeeds if $u$ holds, fails if $v$ holds, chooses demonically between these two possibilities if both $u$ and $v$ hold, and blocks if neither $u$ nor $v$ hold. The {\em assertion} $\{u, v\}$ succeeds if $u$ holds, fails if $v$ holds, choosing angelically between these two possibilities if both $u$ and $v$ hold, and aborts if neither $u$ nor $v$ holds. Neither assumption nor assertion change the state if they succeed or fail.
\[\renewcommand{\arraystretch}{1.2}\begin{array}{rcl@{\qquad\qquad}rcl}
   [u, v]\,(q, r) & \defeq & (u \imp q) \con (v \imp r) & \{u, v\}\,(q, r) & \defeq & (u \con q) \dis (v \con r)
\end{array}\]
Both assumption and assertion are monotonic and hence statements. We have that $[\TRUE, \FALSE] = \SKIP = \{\TRUE, \FALSE\}$ and that $[\FALSE, \TRUE] = \RAISE = \{\FALSE, \TRUE\}$. We also have that $[\FALSE, \FALSE] = \STOP$ and that $\{\FALSE, \FALSE\} = \ABORT$. Finally we have that $[\TRUE, \TRUE] = \SKIP \meet \RAISE$ and that $\{\TRUE, \TRUE\} = \SKIP \join \RAISE$.

The {\em demonic update} $[Q, R]$ and the {\em angelic update} $\{Q, R\}$ both update the state according to relation $Q$ and succeed, or update the state according to relation $R$ and fail, the difference being that both the choice offered by the relations and the choice between succeeding and failing are demonic with $[Q, R]$ and are angelic with $\{Q, R\}$. If $Q$ is of type $\Delta \rel \Psi$ and $R$ is of type $\Delta \rel \Omega$, then $[Q, R]$ and $\{Q, R\}$ are of type ${\cal P} \Psi \times {\cal P} \Omega \to {\cal P} \Delta$:
\[\renewcommand{\arraystretch}{1.2}\begin{array}{lcl}
   [Q, R]\,(q, r) \delta & \defeq & (\forall\, \psi \dot Q\,\delta\,\psi \imp q\,\psi) \con (\forall\, \omega \dot R\,\delta\,\omega \imp r\,\omega) \\
   \{Q, R\}\,(q, r) \delta & \defeq & (\exists\, \psi \dot R\,\delta\,\psi \con q\,\psi) \dis (\exists\, \omega \dot R\,\delta\,\omega \con r\,\omega)
\end{array}\]
Both demonic update and angelic update are monotonic in both arguments and hence statements. Writing $\bot$ for the empty relation and $\ID$ for the identity relation we have that $[\ID, \bot] = \SKIP = \{\ID, \bot\}$ and that $[\bot, \ID] = \RAISE = \{\bot, \ID\}$. We also have that $[\bot, \bot] = \STOP$ and that $\{\bot, \bot\} = \ABORT$. Finally we have that $[\ID, \ID] = \SKIP \meet \RAISE$ and that $\{\ID, \ID\} = \SKIP \join \RAISE$. Writing $\top$ for the universal relation, both updates $[\top, \top]$ and $\{\top, \top\}$ terminate, with $[\top, \top]$ making a demonic choice between succeeding and failing, and a demonic choice among the final states, and $\{\top, \top\}$ making these choices angelic.

\section{Derived Statements}
\label{sec:derived}

To establish the connection to predicate transformers with one argument we define:
\[\renewcommand{\arraystretch}{1.2}\begin{array}{lcl@{\qquad\qquad}lcl}
   [u]\,(q, r) & \defeq & u \imp q & 
     [Q]\,(q, r)\,\delta & \defeq & (\forall\, \psi \dot Q\;\delta\;\psi \imp q\;\psi) \\
   \{u\}\,(q, r) & \defeq & u \con q &
     \{Q\}\,(q, r)\,\delta & \defeq & (\exists\, \psi \dot Q\;\delta\;\psi \con q\;\psi)
\end{array}\]
These definitions are identical to predicate transformers with one argument, except for the additional parameter $r$~\cite{BackVonWright98RefinementCalculus}. We have that $[u] = [u, \FALSE]$ and $\{u\} = \{u, \FALSE\}$ as well as $[Q] = [Q, \bot]$ and $\{Q\} = \{Q, \bot\}$.

The common $\TRY S \CATCH T$ statement with {\em body} $S$ and {\em handler} $T$ starts with $S$, if $S$ succeeds, the whole statement succeeds, if $S$ fails, execution continues with $T$. The try-catch statement is directly defined by exceptional composition. (We have introduced the operator $\SEMI$ to stress the duality to $\semi$ in the algebraic structure; partial refinement will break this duality.) The statement $\TRY S\CATCH T \FINALLY U$ with {\em finalization} $U$ is defined in terms of sequential and exceptional composition: $U$ is executed either after $S$ succeeds, after $S$ fails and $T$ succeeds, or after $S$ fails and $T$ fails, in which case the whole statement fails whether $U$ succeeds or fails:
\[\renewcommand{\arraystretch}{1.2}\begin{array}{lcl}
  \TRY S \CATCH T            & \defeq & S \SEMI T \\
  \TRY S \CATCH T \FINALLY U & \defeq & (S \SEMI (T \SEMI (U \semi \RAISE))) \semi U
\end{array}\]
The {\em assignment} $x := E$ is defined in terms of an update statement that affects only component $x$ of the state space. For this we assume that the state is a tuple and variables select elements of the tuple. Here $E$ may be partially defined, as for example in $x := x \DIV y$. A division by zero should lead to failure without a state change, otherwise to success with $x$ updated. A {\em program expression} $E$ is a term for which {\em definedness} $\DEF E$ (the domain of $E$) and {\em value} $\VAL E$ are given; the result of $\DEF E$ and $\VAL E$ are expressions of the underlying logic of pre- and postconditions. For example, if the state space consists of variables $x, y$, we have $\DEF ``x \DIV y" \equiv (\lambda\,x, y \dot y \neq 0)$ and $\VAL ``x \DIV y" = (\lambda\,x, y \dot \textsf{if}~y \neq 0~\textsf{then}~x~\textsf{div}~y~\textsf{else}~NaN)$. We use upper case names $B, E, ES$ for program expressions (partial functions) and lower case names $e, es, p, q, r, \ldots$ for terms of higher-order logic (total functions). The {\em relational update} $x := e$ modifies the $x$ component of the state space to be $e\,(x, y)$ and leaves all other components of the state space unchanged; the initial and final state space are the same. The {\em nondeterministic relational update} $x :\in es$ modifies the $x$ component of the state space to be any element of the set $es\,(x, y)$. Provided that the state space consists of variables $x, y$ we define:
\[\renewcommand{\arraystretch}{1.2}\begin{array}{lcl}
  x := e   & \defeq & \lambda\,x, y \dot \lambda\,x', y' \dot x' = e\,(x, y) \con y' = y \\
  x :\in es & \defeq & \lambda\,x, y \dot \lambda\,x', y' \dot x' \in es\,(x, y) \con y' = y
\end{array}\]
The (deterministic) \emph{assignment} $x := E$ fails if program expression $E$ is not defined, otherwise it succeeds and assigns the value of $E\,(x, y)$ to $x$. The {\em nondeterministic assignment} $x :\in ES$ fails if the program expression $ES$ is not defined, otherwise it succeeds and assigns any element of the set $ES\,(x, y)$ to $x$, the choice being demonic:
\[\renewcommand{\arraystretch}{1.2}\begin{array}{lcl}
  x := E   & \defeq & [\DEF E, \neg \DEF E] \semi [x := \VAL E] \\
  x :\in ES & \defeq & [\DEF ES, \neg \DEF ES] \semi [x :\in \VAL ES]
\end{array}\]
The {\em conditional} $\IF B \THEN S \ELSE T$ fails if Boolean program expression $B$ is not defined, otherwise continues with either $S$ or $T$, depending on the value of $B$:
\[\renewcommand{\arraystretch}{1.2}\begin{array}{rcl}
\IF B \THEN S \ELSE T & \defeq & [\DEF B, \neg \DEF B] \semi (([\VAL B] \semi S) \meet ([\neg \VAL B] \semi T))
\end{array}\]
The {\em loop} $\WHILE B \DO S$ is defined as the least fixed point of $F = (\lambda\,X \dot \IF B \THEN (S \semi X))$ with respect to the total refinement ordering. According to the Theorem of Knaster-Tarski, any monotonic function $f$ over a complete lattice has a unique least fixed point, written $\mu\,f$. Statements form a complete lattice, sequential composition and the conditional are monotonic, hence we can define for Boolean program expression $B$ and predicate transformer $S$:
\[\WHILE B \DO S \defeq \mu\,(\lambda\,X \dot \IF B \THEN S \semi X \ELSE \SKIP)\]
Loops preserve monotonicity: if $S$ is monotonic, then $\WHILE B \DO S$ is also monotonic.

\section{Total Correctness}
\label{sec:toco}

The total correctness assertion $\TOTAL{p}{S}{q}$ states that under precondition $p$, statement $S$ terminates with postcondition $q$. Correctness assertions give a better intuition than predicate transformers with entailment. Total correctness assertions are  generalized to two postconditions, the normal and exceptional postcondition:
\[
  \TOTAL{p}{S}{q, r} \quad\equiv\quad
    \begin{array}[t]{l}
      \textrm{Under precondition $p$, statement $S$ terminates and } \\
      \textrm{-- on normal termination $q$ holds finally,} \\
      \textrm{-- on exceptional termination $r$ holds finally.}
    \end{array}
\]
We say that under $p$, statement $S$ succeeds with $q$ and fails with $r$. If $\TOTAL{p}{S}{q, \FALSE}$ holds, then $S$ never fails, and we write this more concisely as $\TOTAL{p}{S}{q}$. Let $p, q, r$ be state predicates:
\[\renewcommand{\arraystretch}{1.2}\begin{array}{lcl}
  \TOTAL{p}{S}{q, r} & \defeq & p \leq S \, (q, r) \\
  \TOTAL{p}{S}{q}    & \defeq & p \leq S \, (q, \FALSE)
\end{array}\]
The next theorem summarizes the basic rules for total correctness, considering possible undefinedness of program expressions, see also~\cite{Cristian84CorrectRobustPrograms,Jacobs01FormalisationJavaExceptions, KingMorgan95ExitsInRefinementCalculus, LeinoSnepscheut94SemanticsExceptions}:
\begin{theorem} Let $p, q, r$ be state predicates, $B$ be a Boolean expression, and $S, T$ be statements:
\[\renewcommand{\arraystretch}{1.2}\begin{array}{l@{\quad}l@{\quad}l}
  \TOTAL{p}{\ABORT}{q, r}    & \equiv & p = \FALSE\\
  \TOTAL{p}{\STOP}{q, r}     & \equiv & \TRUE \\
  \TOTAL{p}{\SKIP}{q, r}     & \equiv & p \imp q \\
  \TOTAL{p}{\RAISE}{q, r}    & \equiv & p \imp r \\
  \TOTAL{p}{x := E}{q, r}    & \equiv &
    \begin{array}[t]{@{}l}
      (\DEF E \con p \imp q[x \backslash \VAL E]) \con {} \\
      (\neg \DEF E \con p \imp r)
    \end{array} \\
  \TOTAL{p}{x :\in ES}{q, r}	& \equiv &
    \begin{array}[t]{@{}l}
      (\DEF ES \con p \imp \forall\, x' \in \VAL ES \dot q[x \backslash x']) \con {} \\
      (\neg \DEF ES \con p \imp r)
    \end{array} \\
  \TOTAL{p}{S \semi T}{q, r} & \equiv &
    \begin{array}[t]{@{}l@{}l}
      \exists\, h \dot & \TOTAL{p}{S}{h, r} \con {} \\
                       & \TOTAL{h}{T}{q, r}
    \end{array} \\
  \TOTAL{p}{\TRY S \CATCH T }{q, r} & \equiv &
    \begin{array}[t]{@{}l@{}l}
      \exists\, h \dot & \TOTAL{p}{S}{q, h} \con {} \\
                       & \TOTAL{h}{T}{q, r}
    \end{array} \\
  \TOTAL{p}{S \meet T}{q, r} & \equiv & 
    \begin{array}[t]{@{}l}
      \TOTAL{p}{S}{q, r} \con {} \\
      \TOTAL{p}{T}{q, r} \\
    \end{array} \\
  \TOTAL{p}{\IF B \THEN S \ELSE T}{q, r} & \equiv &
    \begin{array}[t]{@{}l}
      \TOTAL{\DEF B \con \VAL B \con p}{S}{q, r} \con {} \\
      \TOTAL{\DEF B \con \neg \VAL B \con p}{T}{q, r} \con {} \\
      (\neg \DEF B \con p \imp r)		
    \end{array}
\end{array}\]
\end{theorem}
Total correctness and total refinement are related in the same way as for programs with one exit:
\begin{theorem}
\label{thm:total_rel}
For predicate transformers $S, T$:
\[\renewcommand{\arraystretch}{1.2}\begin{array}{rcl}
  S \refby T & \equiv & \forall\, p, q, r \dot \TOTAL{p}{S}{q, r} \imp \TOTAL{p}{T}{q, r}
\end{array}\]
\end{theorem}

\section{Partial Correctness}
\label{sec:paco}

To make the connection between total correctness and partial refinement, we use partial correctness assertion $\PARTIAL{p}{S}{q}$, as introduced in~\cite{SekerinskiZhang11PartialCorrectness}. The notion of total correctness assumes that any possible failure has to be anticipated in the specification; the outcome in case of failure is specified by the exceptional postcondition. An implementation may not fail in any other way. Partial correctness weakens the notion of total correctness by allowing also ``true'' exceptions. For orderly continuation after an unanticipated exception, the restriction is that in that case, the state must not change. We introduce following notation:
\[
  \PARTIAL{p}{S}{q, r} \quad\equiv\quad
    \begin{array}[t]{l}
      \textrm{Under precondition $p$, statement $S$ terminates and } \\
      \textrm{-- on normal termination $q$ holds finally,} \\
      \textrm{-- on exceptional termination $p$ or $r$ holds finally.}
    \end{array}
\]
Both total and partial correctness guarantee termination when the precondition holds. If $\PARTIAL{p}{S}{q, \FALSE}$ holds, then statement $S$ does not modify the state when terminating exceptionally, and we write this more concisely as $\PARTIAL{p}{S}{q}$. Let $p, q, r$ be state predicates:
\[\renewcommand{\arraystretch}{1.2}\begin{array}{lll}
  \PARTIAL{p}{S}{q, r}& \defeq & p \leq S \ (q, p \dis r) \\
  \PARTIAL{p}{S}{q}   & \defeq & p \leq S \ (q, p)
\end{array}\]
Total correctness implies partial correctness, but not vice versa. The very definition of partial correctness breaks the duality between normal and exceptional postconditions that total correctness enjoys. This leads to some curious consequences in the correctness rules.

\begin{theorem} Let $p, q, r$ be state predicates, $B, E$ be program expressions, $x$ be a variable, and $S, T$ be statements:
\[\renewcommand{\arraystretch}{1.2}\begin{array}{l@{\quad}l@{\quad}l}
	\PARTIAL{p}{\ABORT}{q, r}    & \equiv & p = \FALSE \\	
	\PARTIAL{p}{\STOP}{q, r}     & \equiv & \TRUE \\	
	\PARTIAL{p}{\SKIP}{q, r}     & \equiv & p \imp q \\
	\PARTIAL{p}{\RAISE}{q, r}    & \equiv & \TRUE \\
	\PARTIAL{p}{x := E}{q, r}    & \equiv & \DEF E \con p \imp q[x \backslash \VAL E] \\
	\PARTIAL{p}{x :\in ES}{q, r}  & \equiv & \DEF ES \con p \imp \forall\, x' \in \VAL ES \dot q[x \backslash x'] \\
	\PARTIAL{p}{S \semi T}{q, r} & \equiv & \exists\, h \dot
	  \begin{array}[t]{@{}l}
	    \PARTIAL{p}{S}{h, r} \con {} \\
	    \TOTAL{h}{T}{q, p \dis r}
	  \end{array} \\
	\PARTIAL{p}{\TRY S \CATCH T}{q, r} & \equiv & \exists\, h \dot
	  \begin{array}[t]{@{}l}
	    \TOTAL{p}{S}{q, h} \con {} \\
	    \TOTAL{h}{T}{q, p \dis r}
	  \end{array} \\
  \PARTIAL{p}{S \meet T}{q, r} & \equiv &
    \begin{array}[t]{@{}l}
      \PARTIAL{p}{S}{q, r} \con {} \\
      \PARTIAL{p}{T}{q, r}
    \end{array} \\
  \PARTIAL{p}{\IF B \THEN S \ELSE T}{q, r} & \equiv &
    \begin{array}[t]{@{}l}
      \TOTAL{\DEF B \con \VAL B \con p}{S}{q, p \dis r} \con {} \\
      \TOTAL{\DEF B \con \neg \VAL B \con p}{T}{q, p \dis r}
    \end{array}
\end{array}\]
\end{theorem}
The $\RAISE$ statement miraculously satisfies any partial correctness specification by failing and leaving the state unchanged. The rules for assignment and nondeterministic assignment have only conditions in case $E$ is defined; in case $E$ is undefined, the assignment fails without changing the state, thus satisfies the partial correctness specification automatically. We immediately get following consequence rule for any statement $S$:
\[
  \PARTIAL{p}{S}{q, r} \con (q \leq q') \con (r \leq r') \quad\imp\quad \PARTIAL{p}{S}{q', r'}
\]
This rule allows the postconditions to be weakened, like for total correctness, but does not allow the precondition to be weakened.

For $S \semi T$ and $\TRY S \CATCH T$ let us consider the special case when $r = \FALSE$:
\[\renewcommand{\arraystretch}{1.2}\begin{array}{l@{\quad}l@{\quad}l}
	\PARTIAL{p}{S \semi T}{q} & \equiv & \exists\, h \dot
	  \begin{array}[t]{@{}l}
	    \PARTIAL{p}{S}{h} \con 
	    \TOTAL{h}{T}{q, p}
	  \end{array} \\
	\PARTIAL{p}{\TRY S \CATCH T}{q} & \equiv & \exists\, h \dot
	  \begin{array}[t]{@{}l}
	    \TOTAL{p}{S}{q, h} \con 
	    \TOTAL{h}{T}{q, p}
	  \end{array} \\
\end{array}\]
The partial correctness assertion for $S \semi T$ is satisfied if $S$ fails without changing the state, but if $S$ succeeds with $h$, then $T$ must either succeed with the specified postcondition $q$, or fail with the original precondition $p$. For the partial correctness assertion of $\TRY S \CATCH T$ to hold, either $S$ must succeed with the specified postcondition $q$, or fail with $h$, from which $T$ either succeeds with $q$ or fails with the original precondition $p$.

\section{Loop Theorems}
\label{sec:loop}

Let $W \neq \emptyset$ be a well-founded set, i.e.~a set in which there are no infinitely decreasing chains, and let $p_w$ for $w \in W$ be an indexed collection of state predicates called \emph{ranked predicates} of the form $p_w = (p \con (\lambda \sigma \dot v\,\sigma = w)$. Here $p$ is the \emph{invariant} and $v$ the \emph{variant} (\emph{rank}). We define $p_{<w} = (\lambda \sigma \dot (\exists\,w' \in W \dot w' < w \con p_{w'}\,\sigma))$ to be true if a state predicate with lower rank than $p_w$ is true. The fundamental rules for total and partial correctness of loops are as follows:

\begin{theorem}
\label{thm:while_corr}
Assume that $B$ is a Boolean program expression, $r$ is a state predicate and $S$ is a statement. Assume that $p_w$ for $w \in W$ is a ranked collection of state predicates. Then
$$
\begin{array}{rl}
 & \forall\, w \in W \cdot \TOTAL{p_w \con \DEF B \con B}{S}{p_{<w}, r}\\
\imp & \TOTAL{p}{\WHILE B \DO S}{p \con \DEF B \con \neg B, (p \con \neg \DEF B) \dis r}
\end{array}$$
and
$$
\begin{array}{rl}
 & \forall\, w \in W \cdot \PARTIAL{p_w \con \DEF B \con B}{S}{p_{<w}, r}\\
\imp & \PARTIAL{p}{\WHILE B \DO S}{p \con \DEF B \con \neg B, r}
\end{array}$$
\end{theorem}
These theorems separate the concerns of the two exits: on the normal exit, both rules are similar as with one exit, except that the postconditions include the definedness of the guard $B$. On the exceptional exit, the rule for total correctness is the same as its counterpart with one exit~\cite{BackVonWright98RefinementCalculus}; on the exceptional exit, it states that if the loop body $S$ exits exceptionally with postcondition $r$, then the loop exits exceptionally with postcondition $p \con \neg \DEF B$ (failure of the guard) or $r$ (failure of the loop body).

\section{Partial Refinement}
\label{sec:partial}
In this section, we relax total refinement to partial refinement and study its application. Partial refinement of predicate transformers $S, T$ is defined by:
\begin{eqnarr}
  S \prefby T & \defeq & S \sqcap \RAISE \trefby T
\end{eqnarr}%
This implies that on normal exit, $T$ can only do what $S$ does. However, $T$ may fail when $S$ does not, but then has to preserve the initial state. For example, in such a case $T$ would still maintain an invariant and signal to the caller the failure. Note that the types of $S$ and $T$ require the initial state space to be the same as the final state space on exceptional exit. Partial refinement is related  to partial correctness in the same way as total refinement is related to total correcntess:

\begin{theorem}
\label{thm:partial_rel}
For predicate transformers $S, T$:
\[\renewcommand{\arraystretch}{1.2}\begin{array}{rcl}
S \prefby T & \eq & \forall\, p, q, r \dot \PARTIAL{p}{S}{q, r} \imp \PARTIAL{p}{T}{q, r}
\end{array}\]
\end{theorem}
\begin{proof}
Unfolding the definitions yields:
\[(\forall\, q, r \dot (S\,(q, r) \con r) \leq T\,(q, r)) \eq (\forall\, p', q', r' \dot p' \leq S\,(q', p' \dis r') \imp p' \leq T\,(q', p' \dis r'))\]
This is shown by mutual implication. For any $p'$, $q'$ and $r'$, by letting $q$ and $r$ to be $q'$ and $p' \dis r'$ respectively, it is straightforward that the left side implies the right side. Implication of the other direction is more invovled: for any $p$ and $q$, letting $p'$, $q'$ and $r'$ to be $S\,(q, r) \con r$, $q$ and $r$ respectively gives us $(S\,(q, r) \con r) \leq S\,(q, S\,(q, r) \con r \dis r) \imp (S\,(q, r) \con r) \leq T\,(q, S\,(q, r) \con r \dis r)$. Since $S\,(q, r) \con r \dis r = r$, we have $(S\,(q, r) \con r) \leq S\,(q, r) \imp (S\,(q, r) \con r) \leq T\,(q, r)$, which reduces to $S\,(q, r) \con r \leq T\,(q, r)$, thus the left side is implied.
\end{proof}

Total refinement implies partial refinement (in the same way as total correctness implies partial correctness); in addition to $\STOP$ as top element, partial refinement has also $\RAISE$ as top element:
\begin{theorem} For predicate transformers $S, T$:
\begin{eqnarr}
    & S \prefby \RAISE \\
    & S \refby T  \quad \imp & S \prefby T
\end{eqnarr}
\end{theorem}
The fact that $S \prefby \RAISE$ may be surprising, but it does allow for intentionally missing features, the second source of partiality discussed earlier on.
Like total refinement, partial refinement is a preorder, as it is reflexive and transitive:
\begin{theorem} For predicate transformers $S, T, U$:
\begin{eqnarr}
    & S \prefby S \\
    & S \prefby T \con T \prefby U & \imp & S \prefby U
\end{eqnarr}
\end{theorem}
Partial refinement is not antisymmetric, for example $\STOP \prefby \RAISE$ and $\RAISE \prefby \STOP$, but $\STOP \neq \RAISE$. With respect to partial refinement, sequential composition is monotonic only in its first operand, 
while demonic choice and conditional statement are monotonic in both operands:
\begin{theorem}
For statements $S, S', T$:
\begin{eqnarr}
   S \prefby S'                   & \imp & S; T \prefby S'; T \\
   S \prefby S' \con T \prefby T' & \imp & S \sqcap T \prefby S' \sqcap T'\\
                                  & \con & \IF B \THEN S \ELSE T \prefby \IF B \THEN S' \ELSE T'\\
\end{eqnarr}
\end{theorem}
However, $S \prefby S'$ does not imply $T \semi S \prefby T \semi S'$ in general, since $T$ might modify the initial state on normal exit. Similarly, $S \prefby S'$ implies neither $S \SEMI T \prefby S' \SEMI T$ nor $T \SEMI S \prefby T \SEMI S'$.

For \emph{data refinement} we extend the refinement relationships to allow two predicate transformers on possibly different state spaces. We extend the definition of data refinement with predicate transformer, e.g.~\cite{GardinerMorgan91DataRefinementPredicateTransformer,vonWright94LatticeDataRefinement}, which uses an representation operation to link concrete and abstract spaces, to two postconditions.
Suppose $S : {\cal P} \Sigma \times {\cal P} \Sigma \fun {\cal P} \Sigma$ and $T : {\cal P} \Delta \times {\cal P} \Delta \fun {\cal P} \Delta$ are predicate transformers on $\Sigma$ and $\Delta$ respectively, and $R : \Sigma \leftrightarrow \Delta$ is the abstraction relation from $\Sigma$ to $\Delta$. Then we define:
\begin{eqnarr}
T_R & \defeq & [R] \semi ((T \semi \{R^{-1}\}) \SEMI \{\bot, R^{-1}\})
\end{eqnarr}%
The composition $(T \semi \{R^{-1}\}) \SEMI \{\bot, R^{-1}\}$ applies the angelic update $\{R^{-1}\}$ to the normal outcome of $T$, terminating normally, and applies $\{R^{-1}\}$ to the exceptional outcome of $T$, terminating exceptionally. This can be equivalently expressed as $(T \SEMI \{\bot, R^{-1}\}) \semi \{R^{-1}\}$. The demonic update $[R]$ maps the states in $\Sigma$ to states in $\Delta$. Then, after executing $T$, $\{R^{-1}\}$ and $\{\bot, R^{-1}\}$ map the states back to $\Sigma$ from states in $\Delta$, on each exit. In other words, $T_R$ is the projection of $T$ on ${\cal P} \Sigma \times {\cal P} \Sigma \fun {\cal P} \Sigma$ through relation $R$. Now we can define \emph{total data refinement} and {\em partial data refinement} between $S$ and $T$ through $R$ as:
\begin{eqnarr}
S \refby_R T  & \defeq & S \refby T_R \\
S \prefby_R T & \defeq & S \prefby T_R
\end{eqnarr}%
Note that here we could not define $S \prefby_R T$ as $(S; [R]) ;; [R] \prefby [R]; T$, due to the restriction of $\prefby$ that the on both sides the initial state space must be the same as the exceptional state space, which $[R]; T$ obviously does not satisfy.

\subsubsection*{Example: Limitation in Class Implementation}
Now let us revisit the introductory example of the class $IntCollection$. We consider an implementation using a fixed-size array. Using dynamic arrays or a heap-allocated linked list would be treated similarly, as an extension of a dynamic array and heap allocation may fail in the same way as a fixed-sized array may overflow. Since representing dynamic arrays or heaps complicates the model, we illustrate the refinement step using fixed-sized arrays. Let $SIZE$ be a constant of type $\Int$:

\begin{center}
$\begin{array}{l}
\CLASS~IntCollection1 \\
\quad \Int l, \Int []\,a \\
\quad \INVARIANT~inv1: 0 \leq l \leq SIZE \con len(a) = SIZE \con {} \\
\quad \quad (\forall\, x \cdot 0 \leq x < l \imp MIN \leq a[x] \leq MAX) \\
\quad \METHOD~init1() \\
\quad \quad l := 0 \semi a := \NEW\,\Int[SIZE]; \\
\quad \METHOD~insert1(n : \Int) \\
\quad \quad l, a[l] := l + 1, n; \\
\quad \METHOD~sum1(): \Int \\
\quad \quad \Int s, i := 0, 0 ;\\
\quad \quad \{\LOOPINV~linv: 0 \leq i \leq l \con {}\\
\quad \quad \quad s = \sum x \in [0..i) \cdot a[x] \con MIN \leq s \leq MAX)\}\\
\quad \quad \WHILE~i < l~\DO\\
\quad \quad \quad s, i := s + a[i], i + 1 ; \\
\quad \quad \{s = \sum x \in 0..l - 1 \cdot a[x] \con MIN \leq s \leq MAX, inv1\}\\
\quad \quad \RESULT := s\\
\end{array}$
\end{center}
The state spaces of classes $IntCollection$ and $IntCollection1$ are ${\sf bag}(\Int)$ and $\Int \times \Int[]$ respectively. The invariant that links the state spaces is $b = bagof(a[0..l - 1])$, where $bagof$ converts an array to a bag with the same elements. The statement $a := \NEW\,\Int[SIZE]$ might fail due to heap allocation failure. Writing $s^n$ for $n$ times repeating sequence $s$, allocation of an integer array is define as:
\begin{eqnarr}
  x := \NEW\,\Int[n] & \defeq & [\TRUE, \TRUE] \semi x := [0]^n
\end{eqnarr}%
where $[\TRUE, \TRUE]$ might succeed or fail. Here the three methods refine those in class $IntCollection$ respectively, formally $init \prefby_{rel} init1$, $insert \prefby_{rel} insert1$, and $sum \prefby_{rel1} sum1$, in which the relations are given as $rel\,b\,(l, a) \equiv b = bagof(a[0..l - 1])$ and $rel1\,(b, \RESULT)\,(l, a, \RESULT') \equiv b = bagof(a[0..l - 1]) \con \RESULT = \RESULT'$, since $\RESULT$ is part of the state space in $sum1$. We only give the proof sketch of the third one. Using $body$ and $loop$ as the abbreviations of $s, i := s + a[i], i + 1$ and $\WHILE i < l \DO s, i := s + a[i], i + 1$ respectively, and defining $B = (\lambda (s, i, l, a, \RESULT) \dot i < l)$, $linv_w = (linv \con (\lambda (s, i, l, a, \RESULT) \dot w = l - i))$, we have $\DEF B = \TRUE$ and: 

\[\TOTAL{linv \con \DEF B \con B}{body}{linv_{<w}, linv}\]
According to Theorem~\ref{thm:while_corr} we have 

\[\TOTAL{linv \con (\lambda (s, i, l, a, \RESULT) \dot s = 0 \con i = 0)}{loop}{linv \con \DEF B \con \neg B, linv}\]
and $linv \con \DEF B\con \neg B \imp s = \sum x \in 0..l - 1 \cdot a[x]$. We require that the exceptional postconditions must be independent of $\RESULT$ since no value will be returned in failures. With the correctness rule for sequential composition we know that for arbitrary $q$, $r$:

\[\renewcommand{\arraystretch}{1.2}\begin{array}{l}
\TOTAL{(\lambda (l, a, \RESULT) \dot q\,(l, a, \sum x \in 0..l - 1 \cdot a[x])) \con r}{sum1}{q, r}
\end{array}\]
Since for arbitrary $q'$, $r'$,

\[\renewcommand{\arraystretch}{1.2}\begin{array}{rcl}
sum\,(q', r') & = & (\lambda b, \RESULT \dot q'\,(\sum x \in b \cdot x, b) \con MIN \leq \sum x \in b \cdot x \leq MAX) \con r'
\end{array}\]
by definition we know that $sum \prefby_{rel1} sum1$.

Furthermore, $inv1 \, (l, a) \con rel\,b\,(l, a) \imp inv \, b$, which means that the original invariant $inv$ is preserved by the new invariant $inv1$ through relation $rel$. In $sum1$, local variables $s$ and $i$ would be erased on both exits, thus in exceptional termination caused by arithmetic overflow inside the loop, the original state $(l, a)$ remains the same, maintaining the invariant on exceptional exit.

\subsubsection*{Example: Incremental Development}

Another use of partial refinement is to express incremental development. When each of two programs handles some cases of the same specification, then their combination can handle no less cases, while remaining a partial refinement of the specification:
\begin{eqnarr}
     & S \prefby \IF B \THEN T \ELSE \RAISE ~\con~ S \prefby \IF B' \THEN T' \ELSE \RAISE \\
\imp & S \prefby \IF B \THEN T \ELSE (\IF B' \THEN T' \ELSE \RAISE)
\end{eqnarr}%
Moreover, writing $\dis_c$ for \emph{conditional disjunction} (conditional or), defined by $\DEF (B \dis_c B') = \DEF B \con (\VAL B \dis \DEF B')$ and $\VAL (B \dis_c B') = \VAL B \dis (\DEF B \con \VAL B')$, we have
\begin{eqnarr}
\IF B \THEN T \ELSE (\IF B' \THEN T' \ELSE \RAISE) & = & \IF B \dis_c B' \THEN (\IF B \THEN T \ELSE T') \ELSE \RAISE
\end{eqnarr}%
which allows the combination of more than two such programs in a switch-case style, since the right-hand side is again in the ``$\IF \ldots \THEN \ldots \ELSE \RAISE$" form. With more and more such partial implementations combined this way, ideally more cases can be handled incrementally.

\section{Conclusion}
\label{sec:conclusion}


We introduced partial refinement for the description for programs that are unable to fully implement their specifications. Using predicate transformers with two arguments for the semantics of statements allows us to specify the behaviour on both normal and exceptional exits. The partial refinement relation suggests that programs should complete the task as expected on the normal exit, but when they fail and terminate on the exceptional exit, restore the (abstract) initial state. We have not addressed a number of issue, e.g.~rules for partial refinement of loops and recursive procedures (the example of $IntCollection1$ proves refinement on predicate level). King and Morgan as well as Watson present a number of rules for total refinement, although using different programming constructs~\cite{KingMorgan95ExitsInRefinementCalculus,Watson02RefiningExceptions}: $S > T$, pronounced ``$S$ else $T$", can be defined here as $S \meet (T \semi \RAISE)$, the \emph{specification statement} $x:[p, q > r]$ with \emph{frame} $x$ can be defined here as $\{p\} \semi [\overline{q}, \overline{r}]$, where $\overline{q}, \overline{r}$ lifts predicates $q, r$ to relations over $x$, the \emph{exception block} $[\hspace{-.15em}[S]\hspace{-.15em}]$ with $\RAISE H$ in its body, where $H$ is the exception handler, can be defined here as $S \SEMI H$. This allows their total refinement rules to be used here. However, the purpose of partial refinement is different and we would expect different kinds of rules needed.

For two sources of partiality, inherent limitations of an implementation and intentionally missing implementation, the two examples in Sec.~\ref{sec:partial} illustrate how partial refinement can be used to help in reasoning. For the third source of partiality, genuine faults, the approach needs further elaboration. In earlier work, we have applied the notions of partial correctness to three design patterns for fault tolerance, rollback, degraded service, and recovery block~\cite{SekerinskiZhang11PartialCorrectness}. We believe that these ideas can be carried over to partial refinement.

Implementation restrictions are also addressed by IO-refinement, which allows the type of method parameters to be changed in refinement steps~\cite{BoitenDerrick98IORefinement, BanachaPoppletonJeskeaStepney07Retrenchment, DerrickBoiten01RefinementZObjectZ, MikhajlovaSekerinski97ClassInterfaceRefinement, StepneyCooperWoodcock98ZDataRefinement}. Partial refinement and IO-refinement are independent of each other and can be combined.

Data refinement as used here is a generalization of \emph{forward data refinement} to two exits. Forward data refinement is known to be incomplete; for the predicate transformer model of statements, several alternatives have been studied, e.g.~\cite{GardinerMorgan93RuleDataRefinement, vonWright94LatticeDataRefinement}. It remains to be explored how these can be used for partial data refinement.

An alternative to using predicate transformers with two arguments is to encode the exit into the state and represent programs as predicate transformers with one argument, which we have not pursued. The advantage of predicate transformers with two arguments is that the state spaces on normal and exceptional exit can be different, which is useful for local variable declarations that can either fail and leave the state space unchanged or succeed and enlarge the state space. Having two separate postconditions may be methodologically stronger and syntactically shorter.

While predicate transformers allow to describe angelic nondeterminism, we have not made use of that. Finally, we note that partial refinement relies only a single exceptional exit. Common programming languages provide named exceptions and statements with several exits. Exploring a generalization of partial refinement to multiple exits is left as future work.

\paragraph{Acknowledgements.} We are grateful to the reviewers, whose comments were exceptionally detailed and constructive.

\bibliographystyle{eptcs}
\bibliography{exceptionref}

\end{document}